\documentclass{article}
\pdfoutput=1
\usepackage{mathptmx}
\usepackage{helvet}
\usepackage{courier}
\usepackage[T1]{fontenc}
\usepackage[latin9]{inputenc}
\usepackage{array}
\usepackage{float}
\usepackage{booktabs}
\usepackage{multirow}
\usepackage{amsmath}
\usepackage{amsthm}
\usepackage{amssymb}
\usepackage{graphicx}
\usepackage{esint}
\usepackage[authoryear]{natbib}
\usepackage[unicode=true,
 bookmarks=false,
 breaklinks=false,pdfborder={0 0 1},backref=false,colorlinks=false]
 {hyperref}
\hypersetup{pdftitle={Bell's theorem in automata theory},
 pdfauthor={Michael Zirpel},
 pdfsubject={An automata-theoretic version of Bell's theorem},
 pdfkeywords={Bell's theorem, Bell test experiment, Bell-CHSH inequality, probabilistic sequential machines, quantum sequential machines}}

\makeatletter

\let\SF@@footnote\footnote
\def\footnote{\ifx\protect\@typeset@protect
    \expandafter\SF@@footnote
  \else
    \expandafter\SF@gobble@opt
  \fi
}
\expandafter\def\csname SF@gobble@opt \endcsname{\@ifnextchar[%]
  \SF@gobble@twobracket
  \@gobble
}
\edef\SF@gobble@opt{\noexpand\protect
  \expandafter\noexpand\csname SF@gobble@opt \endcsname}
\def\SF@gobble@twobracket[#1]#2{}
\providecommand{\tabularnewline}{\\}

\numberwithin{equation}{section}
\numberwithin{figure}{section}
\theoremstyle{plain}
\newtheorem{thm}{\protect\theoremname}
  \theoremstyle{definition}
  \newtheorem{example}[thm]{\protect\examplename}
  \theoremstyle{plain}
  \newtheorem*{prop*}{\protect\propositionname}

\@ifundefined{date}{}{\date{}}
\makeatother

  \providecommand{\examplename}{Example}
  \providecommand{\propositionname}{Proposition}
\providecommand{\theoremname}{Theorem}

\begin{document}

\title{\noindent Bell's theorem in automata theory}

\author{Michael Zirpel, mz@mzirpel.de}
\maketitle
\begin{abstract}
Bell's theorem is reformulated and proved in the pure mathematical
terms of automata theory, avoiding any physical or ontological notions.
It is stated that no pair of finite probabilistic sequential machines
can reproduce in its output the statistical results of the quantum-physical
Bell test experiment if each machine is independent of the respective
remote input. 

\end{abstract}

\section{Introduction}

Bell's theorem \citep{Bell1964} was praised as one of the most profound
discoveries of science \citep{Stapp75}. Even more than 50 years after
its discovery there is a vivid discussion of its meaning and its impact
in a plenty of scientific papers. And the thought experiment on which
the theorem is based, the \emph{Bell test} experiment, is performed
each year in new variants (e.g., \href{https://thebigbelltest.org}{https://thebigbelltest.org}). 

Bell's theorem states that ``no physical theory which is realistic
and also local in a specified sense can agree with all of the statistical
implications of Quantum Mechanics'' \citep{Shimony2013}. However,
the meaning of this proposition is not easy to understand, neither
its consequences for cryptographic protocols (e.g., \citealp{Ekert1991}). 

In the following the theorem will be reformulated for automata theory
in a pure mathematical way\footnote{The idea to use computer or electronic circuits to explain the content
of Bell's theorem is not new (e.g., \citealp{Gill2014}) and inspired
this paper. But we consider abstract mathematical automata instead
of real physical devices. } without any physical or ontological notions. For that purpose the
arrangement of an ideal Bell test experiment is represented by a pair
of automata, deterministic or probabilistic sequential machines, which
produce an output after each input given by local operators or independent
random generators. The theorem states that no such pair can reproduce
the statistical results of the quantum-physical Bell test experiment
in its output data if each machine is independent on the input of
the respective remote machine, or loosely speaking: without data transmission
between the remote sides.

After presenting a short sketch of the physical Bell test experiment,
we will give a brief introduction to the theory of sequential machines
and then prove the theorem.

\section{Bell test experiment}

The fundamental idea of the Bell test experiment has a long lasting
history: Einstein, Podolsky, and Rosen (EPR \citeyear{EPR1935}) developed
a quantum-physical thought experiment that displayed strange non-local
correlations between the results of remote measurements on a pair
of particles, depending on the choice of the measured quantity. The
setting of this thought experiment was simplified by \citet{Bohm1951}
and inspired \citet{Bell1964} to his theorem. Experimentalists like
Clauser, Horne, Shimony, and Holt (CHSH \citeyear{CHSH1969}) transformed
Bell's thought experiment into a real one, using photon pairs, with
the first sufficient realization by \citet{Aspect1982}. 

For our purpose a coarse sketch of an ideal Bell test experiment without
any physical details is sufficient.\footnote{A good introduction to the physical thought experiment and the theorem
is given in \citet{Bell1981}. The SEP article \citep{Shimony2013}
is closer to our considerations.} 

\noindent \begin{center}
\begin{figure}[H]
\noindent \begin{centering}
\includegraphics{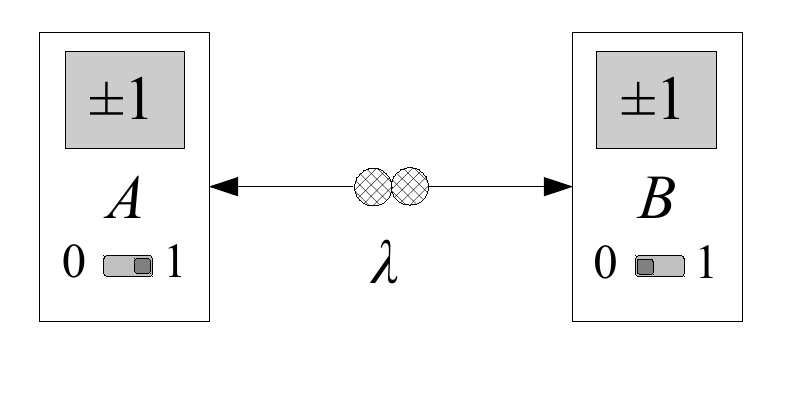}
\par\end{centering}
\caption{\label{BellTestExp}Bell test experiment}
\end{figure}
\par\end{center}

A single run of the experiment starts by sending a pair of particles
$\lambda$ from a source, each particle in another direction, to the
measurement devices $A$ and $B$ (fig.\ \ref{BellTestExp}). Each
measurement device has an input switch with two positions $0$, $1$,
where local operators or independent random generators can select
one of two different measurements to be performed. So on each side
one of two possible quantities $A_{0}$, $A_{1}$, respectively $B_{0}$,
$B_{1}$, is measured. The display of the apparatus shows the measurement
result, which is in all cases $-1$ or $1$.

This arrangement enables the measurement of one of four possible pairs
$(A_{0},B_{0})$, $(A_{0},B_{1})$, $(A_{1},B_{0})$, $(A_{1},B_{1})$
in each run of the experiment, from which the corresponding product
$A_{0}B_{0}$, $A_{0}B_{1}$, $A_{1}B_{0}$, or $A_{1}B_{1}$ is computed,
which has the value $-1$ or $1$. The procedure will be repeated
with different random measurement selections. After several runs the
mean values of the products $\overline{A_{0}B_{0}}$, $\overline{A_{0}B_{1}}$,
$\overline{A_{1}B_{0}}$, $\overline{A_{1}B_{1}}$ are used to compute
the following expression
\[
\overline{A_{0}B_{0}}+\overline{A_{0}B_{1}}+\overline{A_{1}B_{0}}-\overline{A_{1}B_{1}}\textrm{,}
\]
which in the long run should approximate the theoretical given expectation
value 
\[
\bigl\langle E_{\mathrm{CHSH}}\bigr\rangle=\bigl\langle A_{0}B_{0}\bigr\rangle+\bigl\langle A_{0}B_{1}\bigr\rangle+\bigl\langle A_{1}B_{0}\bigr\rangle-\bigl\langle A_{1}B_{1}\bigr\rangle=\bigl\langle A_{0}B_{0}+A_{0}B_{1}+A_{1}B_{0}-A_{1}B_{1}\bigr\rangle\textrm{.}
\]

In probability theory for any four random variables $A_{0}$, $A_{1}$,
$B_{0}$, $B_{1}$ with the image $\{-1,1\}$ on an event space $(\Omega,\mathcal{A})$
the absolute value of this expectation value is bounded according
the \emph{Bell-CHSH inequality} \citep{CHSH1969} by the value $2$
(cf.\ app.\ A), so for any probability measure $\mu$ on $(\Omega,\mathcal{A})$
\[
\bigl|\bigl\langle E_{\mathrm{CHSH}}\bigr\rangle\bigr|\leq2\textrm{.}
\]

However, quantum theory predicts in some cases values above $\text{2}$
(and below or equal $2\sqrt{2}$).\footnote{In quantum theory the four measurable quantities are represented by
four self-adjoint operators with the spectrum $\{-1,1\}$ on a Hilbert
space $\mathcal{H}$. The quantum-theoretical expectation value of
the corresponding Bell-CHSH expression has a higher bound, the \emph{Tsirelson
bound} $2\sqrt{2}$ (\citealp{Cirelson1980}).} That was confirmed by the measurement results of various quantum-physical
Bell test experiments (e.g., \citealt{Aspect1982}). So the \emph{violation}
of the Bell-CHSH inequality is an experimental fact of quantum physics.

\section{Probabilistic sequential machines\protect\footnote{This section is based on \citet{Salomaa1969} }}

A \emph{sequential machine} (SM) is an abstract automaton that after
each input produces an output and may change its internal state. $I$,
$O$, $S$ denote the sets of \emph{input symbols}, \emph{output symbols}
(also called input and output alphabet) and \emph{states}. A SM is
called \emph{finite} if all these sets $I$, $O$, $S$ are finite. 

A \emph{deterministic SM} (DSM) is defined by a quintuple $(I,O,S,s_{0},f)$
where $s_{0}\in S$ is the initial state and the \emph{deterministic
machine function} 
\[
f\colon I\times S\rightarrow O\times S,(i,s)\mapsto(o,t)=f(i,s)
\]
determines the output $o$ and the new state $t$ after an input $i$
in state $s$. For a \emph{probabilistic SM} the new state $t$ and
the output symbol $o$ are randomly chosen after each input. 

A \emph{finite probabilistic SM} (FPSM) is defined by a quintuple
$(I,O,S,p_{0},p)$ with finite $I$, $O$, $S$, an \emph{initial
state distribution function} 
\[
p_{0}\colon S\rightarrow[0,1],s\mapsto p_{0}(s)
\]
with 
\[
\sum_{s\in S}p_{0}(s)=1,
\]
and a \emph{probabilistic machine function} 
\[
p\colon O\times S\times I\times S\rightarrow[0,1],(o,t,i,s)\mapsto p(o,t\mid i,s)\textrm{,}
\]
which gives the probability to get the output $o$ and the new state
$t$ after the input $i$ in state $s$, with
\[
\sum_{o\in O}\sum_{t\in S}p(o,t\mid i,s)=1
\]
for all $i\in I,s\in S$. 

A FPSM $(I,O,S,p_{0},p)$ is \emph{deterministic} if the image of
the functions $p_{0}$ and $p$ is $\{0,1\}.$ In that case an\emph{
}equivalent finite DSM $(I,O,S,s_{0},f)$ is defined by the initial
state $s_{0}\in S$ which is uniquely determined by $p_{0}(s_{0})=1$,
and the function $f$ which is given by the set of pairs $((i,s)\mapsto(o,t))$
with $p(o,t\mid i,s)=1$.\footnote{Also the more popular Moore and Mealy machines can be considered as
FPSMs with special forms of the probabilistic machine function $p$
(see \citealp{Salomaa1969}).} 

\section{Simulation of Bell test experiments with FPSMs}

The theory of FPSMs is versatile enough to describe any simulation
of the Bell test experiment with a computer or an electronic circuit.
We start with a single FPSM for the simulation where the output is
a pair $o=(A,B)\in O=\{-1,1\}^{2}$ that represents the measurement
results. The input is given as a triple $i=(a,b,\lambda)\in I=\{0,1\}^{2}\times\varLambda$,
where $a$ and $b$ represent the experimenters choices and $\lambda\in\varLambda$
some not further specified properties of the particle pair, which
can be used like the internal states $S$ for a computational model
of the experiment. The simulation FPSM is denoted by 
\[
M=(\{0,1\}^{2}\times\varLambda,\{-1,1\}^{2},S,p_{0},p)\textrm{.}
\]
\begin{figure}[h]
\noindent \begin{centering}
\includegraphics[scale=0.5]{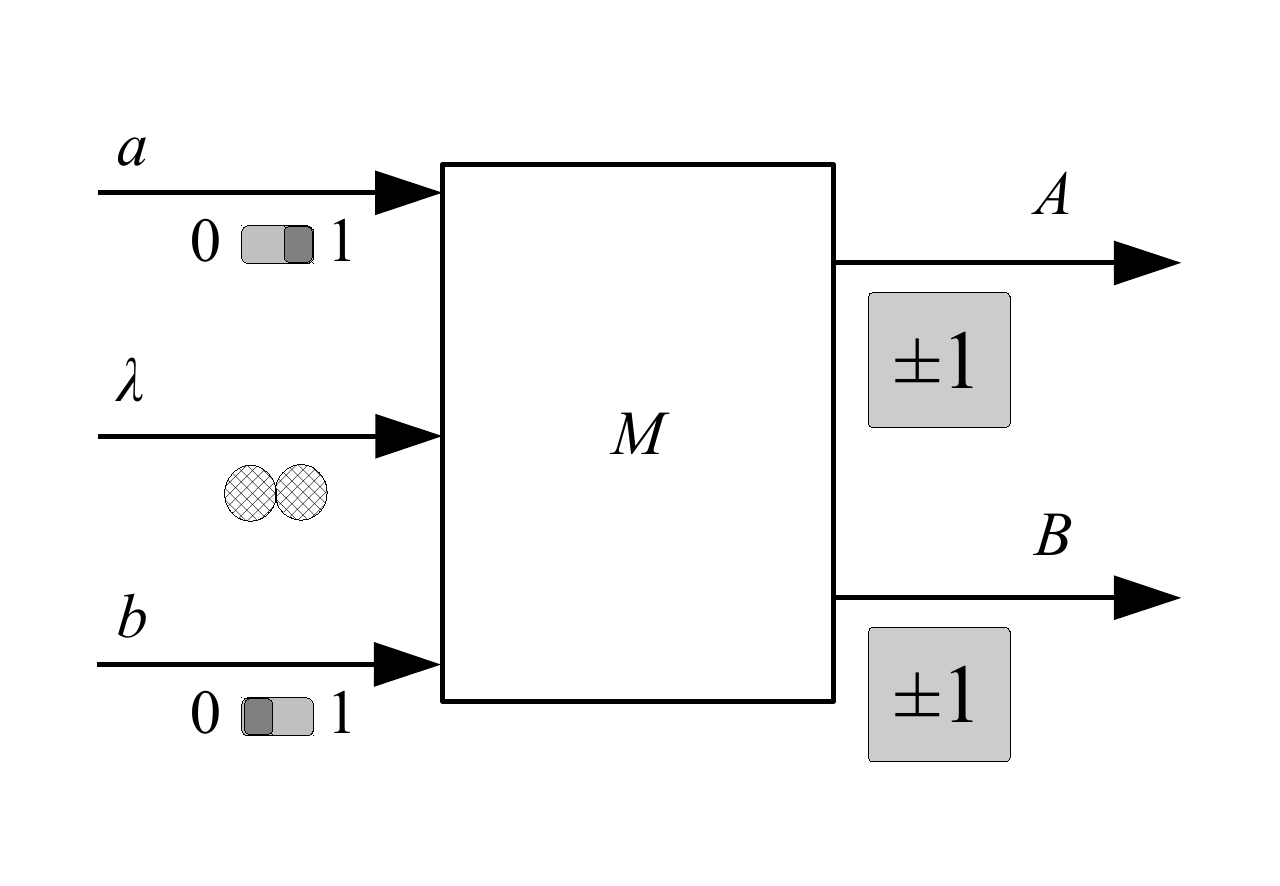}
\par\end{centering}
\caption{Simulation of the Bell test experiment with a FPSM\label{fig:Simulation1FPSM} }
\end{figure}

It is obvious that the FPSM simulation has less constraints then the
original experiment: the outputs may not represent four measurable
quantities $A_{0},A_{1},B_{0},B_{1}$. So we have to use a slightly
more general notation. 

\subsection{Simulation protocol}

In each run of the simulation the FPSM is prepared randomly to an
initial state $s\in S$ according to the probability distribution
$p_{0}$, and an input symbol $\lambda\in\varLambda$ is entered,
randomly selected according to a probability distribution $p_{\varLambda}$.
Furthermore the input symbols $a,b\in\{0,1\}$ are entered by the
operators or automatically by independent random generators. Then
the output symbols $A,B\in\{-1,1\}$ (and the new state $t\in S$)
are given by the machine according to the probabilistic machine function
$p$. 

The product of the output symbols $A\cdot B$ is recorded by the operator
together with the corresponding pair of input symbols $(a,b)$

\[
\left.AB\right.\mid_{a,b}
\]
After a series of runs (or multiple series for the different combinations
of the values of $a$, $b$) the mean values of the recorded products
for the different input symbols are computed

\[
\overline{\left.AB\right.}\mid_{0,0}\textrm{, }\overline{\left.AB\right.}\mid_{0,1}\textrm{, }\overline{\left.AB\right.}\mid_{1,0}\textrm{, }\overline{\left.AB\right.}\mid_{1,1}\textrm{,}
\]
as well as the Bell-CHSH expression

\[
\overline{E}_{\mathrm{CHSH}}=\overline{\left.AB\right.}\mid_{0,0}+\overline{\left.AB\right.}\mid_{0,1}+\overline{\left.AB\right.}\mid_{1,0}-\overline{\left.AB\right.}\mid_{1,1}\textrm{.}
\]
In the special case of a Bell test simulation this value should approximate
in the long run the theoretical given expectation value of the Bell-CHSH
expression

\[
\bigl\langle E_{\textrm{CHSH}}\bigr\rangle=\bigl\langle A_{0}B_{0}\bigr\rangle+\bigl\langle A_{0}B_{1}\bigr\rangle+\bigl\langle A_{1}B_{0}\bigr\rangle-\bigl\langle A_{1}B_{1}\bigr\rangle\textrm{.}
\]
But in general the corresponding theoretical expression for the PSM
is a sum of conditional expectation values
\[
\tilde{E}_{\mathrm{CHSH}}=\bigl\langle AB\bigr\rangle\mid_{0,0}+\bigl\langle AB\bigr\rangle\mid_{0,1}+\bigl\langle AB\bigr\rangle\mid_{1,0}-\bigl\langle AB\bigr\rangle\mid_{1,1}\textrm{,}
\]
which can be computed by
\[
\bigl\langle AB\bigr\rangle\mid_{a,b}=\sum_{A\in\{-1,1\}}\sum_{B\in\{-1,1\}}AB\cdot q(A,B|a,b)
\]
with the conditional output probability to get the output $(A,B)$
after input $(a,b)$ 
\[
q\left(A,B\mid a,b\right)=\sum_{\lambda\in\varLambda}\sum_{s\in S}\sum_{t\in S}p(A,B,t\mid a,b,\lambda,s)p_{\Lambda}(\lambda)p_{0}(s)\textrm{.}
\]
\begin{example}
The following table gives the probabilistic machine functions of several
simple FPSMs for the simulation of the Bell test experiment. 

\noindent {\footnotesize{}}
\begin{table}[H]
\noindent {\footnotesize{}}%
\begin{tabular}{ccccccc}
\toprule 
 & {\footnotesize{}$q(A,B\mid a,b)=p(A,B,t\mid a,b,\lambda,s)$} & {\footnotesize{}$\bigl\langle AB\bigr\rangle|{}_{0,0}$} & {\footnotesize{}$\bigl\langle AB\bigr\rangle|_{0,1}$} & {\footnotesize{}$\bigl\langle AB\bigr\rangle|_{1,0}$} & {\footnotesize{}$\bigl\langle AB\bigr\rangle|_{1,1}$} & {\footnotesize{}$\tilde{E}_{\mathrm{CHSH}}$}\tabularnewline
\midrule
{\footnotesize{}$M_{1}$} & {\footnotesize{}$\frac{1}{4}$} & {\footnotesize{}$0$} & {\footnotesize{}$0$} & {\footnotesize{}$0$} & {\footnotesize{}$0$} & {\footnotesize{}$0$}\tabularnewline
{\footnotesize{}$M_{2}$} & {\footnotesize{}$\delta_{A,1}\delta_{B,1}$} & {\footnotesize{}$1$} & {\footnotesize{}1} & {\footnotesize{}$1$} & {\footnotesize{}$1$} & {\footnotesize{}$2$}\tabularnewline
{\footnotesize{}$M_{3}$} & {\footnotesize{}$\delta_{A,1}\delta_{B,1-2ab}$} & {\footnotesize{}$1$} & {\footnotesize{}1} & {\footnotesize{}$1$} & {\footnotesize{}$-1$} & {\footnotesize{}$4$}\tabularnewline
{\footnotesize{}$M_{4}$} & {\footnotesize{}$\frac{1}{2}\delta_{A,1}\delta_{B,1-2ab}+\frac{1}{2}\delta_{A,-1}\delta_{B,2ab-1}$} & {\footnotesize{}$1$} & {\footnotesize{}1} & {\footnotesize{}$1$} & {\footnotesize{}$-1$} & {\footnotesize{}$4$}\tabularnewline
{\footnotesize{}$M_{5}$} & {\footnotesize{}$\frac{2-\sqrt{2}}{8}+\frac{\sqrt{2}}{8}(2\delta_{A,B}+2ab-4ab\delta_{A,B})$} & {\footnotesize{}$\frac{\sqrt{2}}{2}$} & {\footnotesize{}$\frac{\sqrt{2}}{2}$} & {\footnotesize{}$\frac{\sqrt{2}}{2}$} & {\footnotesize{}$-\frac{\sqrt{2}}{2}$} & {\footnotesize{}$2\sqrt{2}$}\tabularnewline
\bottomrule
\end{tabular}{\footnotesize \par}

{\footnotesize{}\caption{Probabilistic machine functions and expectations of the example FPSMs{\footnotesize{} }}
}{\footnotesize \par}
\end{table}
{\footnotesize \par}

\noindent There is no dependence on internal states or $\lambda$-input,
so we assume $S=\{s_{0}\}$ with $p_{0}(s_{0})=1$ and $\varLambda=\{\lambda_{0}\}$
with $p_{\Lambda}(\lambda_{0})=1$. In this case the probabilistic
machine function is equal to the conditional output probability $q(A,B\mid a,b)=p(A,B,t\mid a,b,\lambda,s)$
for all $\lambda\in\varLambda$; $s,t\in S$. $\delta_{x,y}$ is the
Kronecker symbol and has the value $1$ if $x=y$ and $0$ otherwise.

The output of FPSM $M_{1}$ is evenly distributed and uncorrelated
random, whereas $M_{2}$ gives the constant output $(1,1)$. FPSM
$M_{3}$ modifies the constant output $1$ in the case that $a$ and
$b$ have the value $1.$ FPSM $M_{4}$ is a mixture of $M_{3}$ and
its negative counterpart and simulates a \emph{Popescu-Rohrlich box}
(\emph{PR box}). FPSM $M_{5}$ is the simulation of a quantum-physical
Bell test.

The FPSMs $M_{2}$ and $M_{3}$ are deterministic. The value of $\tilde{E}_{\textrm{CHSH}}$
indicates that the FPSMs $M_{3}$, $M_{4}$, $M_{5}$ violate the
Bell-CHSH inequality.

The free web app \href{https://bell.qlwi.de}{https://bell.qlwi.de}
can be used to perform Bell test simulations with these FPSMs on any
PC, tablet, or smartphone with an up-to-date internet browser.
\end{example}

\subsection{Machine composition and stochastic independence of the machine functions}

The example FPSM $M_{5}$ demonstrates that it is possible to simulate
a Bell experiment with a FPSM and get the same statistical results
as with a quantum-physical experiment.

However, to shed some light on Bell's theorem the simulation has to
be performed with a pair of two separated FPSMs, as sketched by the
circuit diagram in fig.\ \ref{fig:Simulation2FPSMs}). 
\begin{center}
\begin{figure}[h]
\begin{centering}
\includegraphics[scale=0.5]{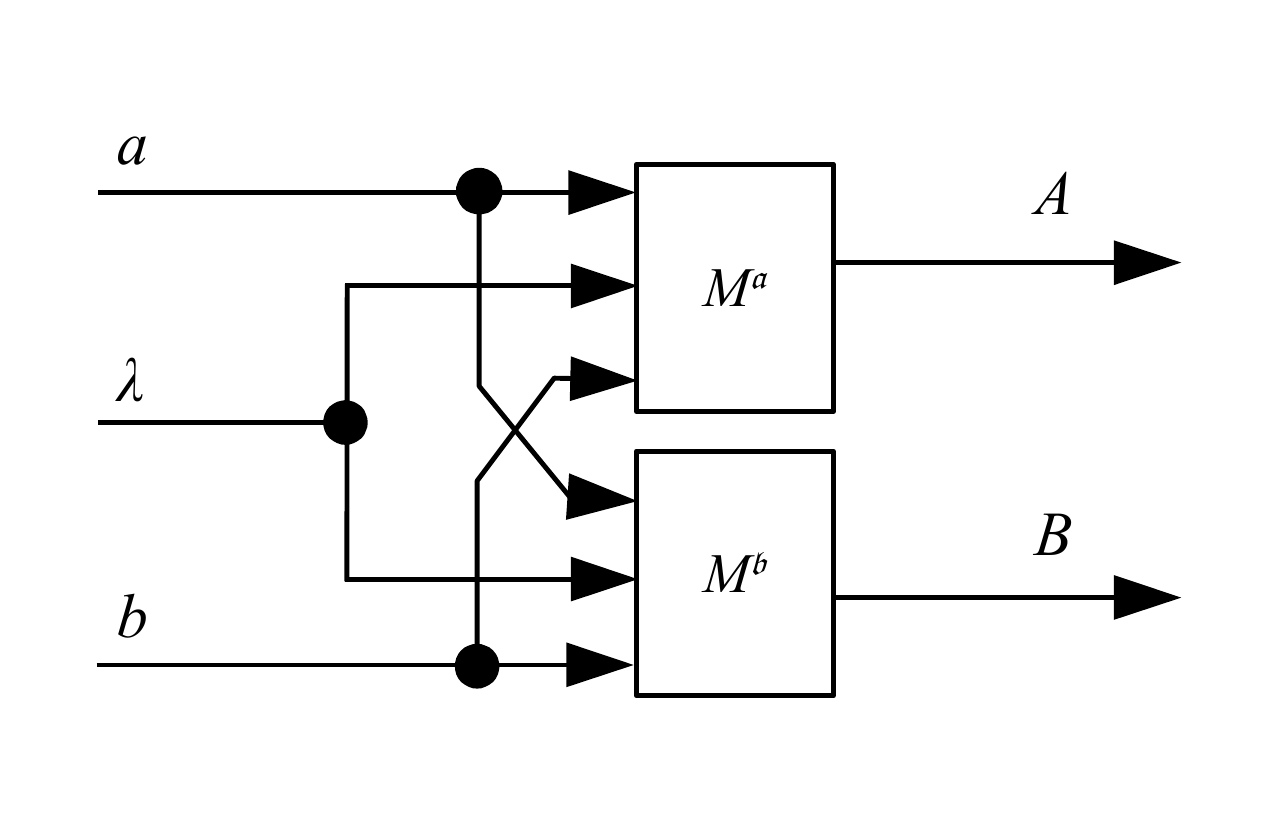}
\par\end{centering}
\caption{Simulation of the Bell test experiment with two FPSMs\label{fig:Simulation2FPSMs}}
\end{figure}
\par\end{center}

Both FPSMs\footnote{We use the letters $\mathfrak{a}$,$\mathfrak{b}$ in the upper position
only as label not as exponent or summation index. } 
\[
M^{\mathfrak{a}}=(\{0,1\}^{2}\times\varLambda,\{-1,1\},S^{\mathfrak{a}},p_{0}^{\mathfrak{a}},p^{\mathfrak{a}})
\]
\[
M^{\mathfrak{b}}=(\{0,1\}^{2}\times\varLambda,\{-1,1\},S^{\mathfrak{b}},p_{0}^{\mathfrak{b}},p^{\mathfrak{b}})
\]
receive the same input (this also ensures synchronization). But each
one gives only one output $A\in\{-1,1\}$, respectively $B\in\{-1,1\}.$
The pair $(M^{\mathfrak{a}},M^{\mathfrak{b}})$ can be considered
as a compound FPSM 
\[
M^{\mathfrak{ab}}=(\{0,1\}^{2}\times\varLambda,\{-1,1\}^{2},S^{\mathfrak{a}}\times S^{\mathfrak{b}},p_{0}^{\mathfrak{a}}p_{0}^{\mathfrak{b}},p^{\mathfrak{a}}p^{\mathfrak{b}})\textrm{,}
\]
where the probabilistic machine function and the initial state distribution
function are given as products of the corresponding functions of the
components, which reflects the \emph{independence} of the machines.
For that reason not every FPSM can be replaced by a pair. 
\begin{example}
The FPSMs $M_{4}$ and $M_{5}$ of example 1 cannot be replaced by
pair, but the FPSMs $M_{1}$, $M_{2}$, and $M_{3}$ can. If we assume
$S^{\mathfrak{a}}=\{s_{0}^{\mathfrak{a}}\}$, $S^{\mathfrak{b}}=\{s_{0}^{\mathfrak{b}}\}$
and $s_{0}=(s_{0}^{\mathfrak{a}},s_{0}^{\mathfrak{b}})$ with $p_{0}^{\mathfrak{a}}(s_{0}^{\mathfrak{a}})=p_{0}^{\mathfrak{b}}(s_{0}^{\mathfrak{b}})=1$,
then $M_{1}=(\{0,1\}^{2}\times\{\lambda_{0}\},\{-1,1\}^{2},\{s_{0}\},p_{1},1)$
can be replaced by a pair $(M_{1}^{\mathfrak{a}},M_{1}^{\mathfrak{b}})$
with the probabilistic machine functions $p_{1}^{\mathfrak{a}}=\frac{1}{2}$
and $p_{1}^{\mathfrak{b}}=\frac{1}{2}$ because $p_{1}=p_{1}^{\mathfrak{a}}p_{1}^{\mathfrak{b}}$.
Similarly, $M_{2}$ can be replaced by $(M_{2}^{\mathfrak{a}},M_{2}^{\mathfrak{b}})$
with $p_{2}^{\mathfrak{a}}=\delta_{A,1}$ and $p_{2}^{\mathfrak{b}}=\delta_{B,1}$,
and $M^{3}$ by $(M_{3}^{\mathfrak{a}},M_{3}^{\mathfrak{b}})$ with
$p_{3}^{\mathfrak{a}}=\delta_{A,1}$ and $p_{3}^{\mathfrak{b}}=\delta_{B,1-2ab}$
. 
\end{example}

\subsection{Functional independence from the remote inputs and the Bell-CHSH
inequality}

Now we consider the case that the machine $M^{\mathfrak{a}}$ does
not depend on input $b$ and the machine $M^{\mathfrak{b}}$ does
not depend on input $a$. In this case the vertical connections can
be removed from the circuit diagram (dotted lines in fig.\ \ref{fig:Simulation2FPSMsNoconn}).
\noindent \begin{center}
\begin{figure}[H]
\noindent \begin{centering}
\includegraphics[scale=0.5]{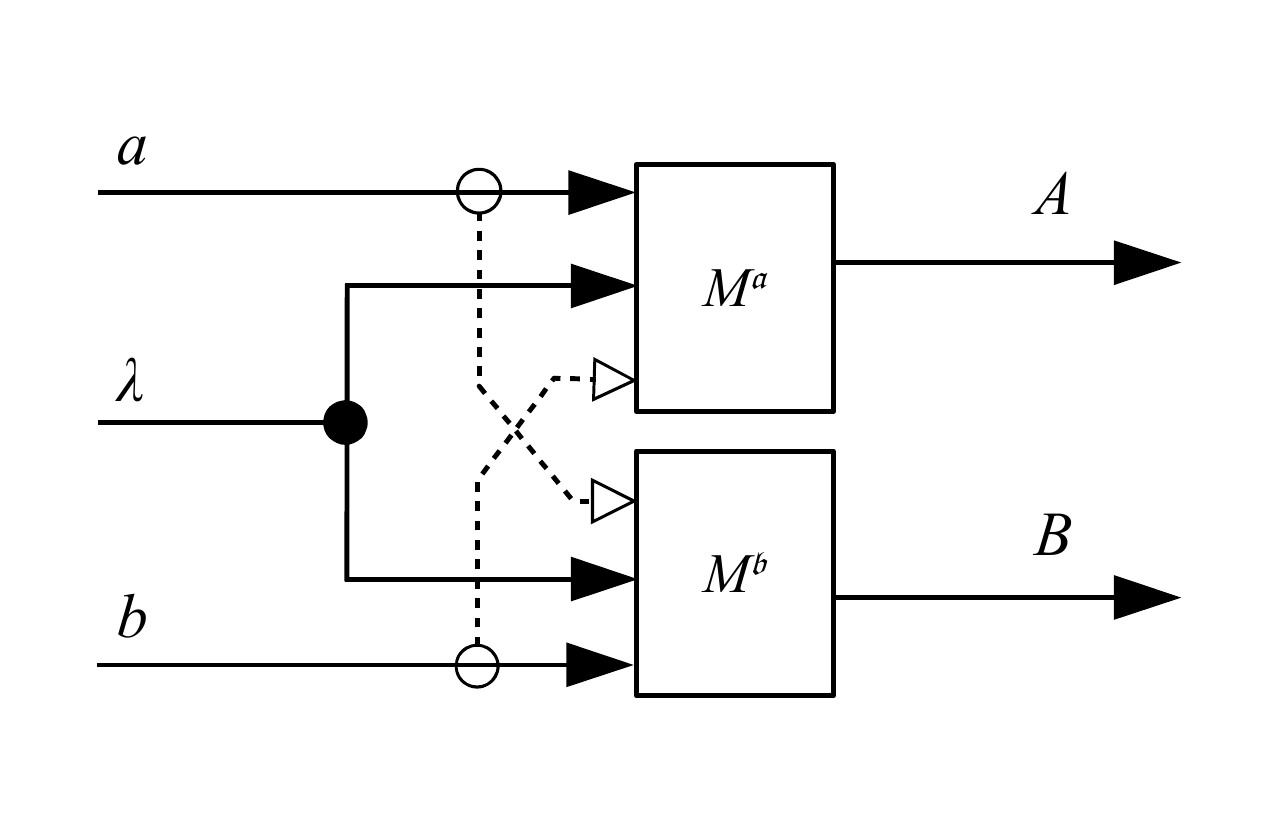}
\par\end{centering}
\caption{Simulation with two FPSMs without dependence on the remote inputs\label{fig:Simulation2FPSMsNoconn}}
\end{figure}
\par\end{center}
\begin{prop*}
If $p^{\mathfrak{a}}$ is not dependent on selection input $b$ and
$p^{\mathfrak{b}}$ is not dependent on selection input $a$, i.e.,
\[
p^{\mathfrak{a}}(A,t^{\mathfrak{a}}\mid a,b,\lambda,s^{\mathfrak{a}})=p^{\mathfrak{a}}(A,t^{\mathfrak{a}}\mid a,0,\lambda,s^{\mathfrak{a}})=p^{\mathfrak{a}}(A,t^{\mathfrak{a}}\mid a,1,\lambda,s^{\mathfrak{a}})
\]
\[
p^{\mathfrak{b}}(B,t^{\mathfrak{b}}\mid a,b,\lambda,s^{\mathfrak{b}})=p^{\mathfrak{b}}(B,t^{\mathfrak{b}}\mid0,b,\lambda,s^{\mathfrak{b}})=p^{\mathfrak{b}}(B,t^{\mathfrak{b}}\mid1,b,\lambda,s^{\mathfrak{b}})
\]
for all $A,B\in\{-1,1\}$; $a,b\in\{0,1\}$; $\lambda\in\varLambda$;$s^{\mathfrak{a}},t^{\mathfrak{a}}\in S^{\mathfrak{a}};s^{\mathfrak{b}},t^{\mathfrak{b}}\in S^{\mathfrak{b}}$,
then the Bell-CHSH expression will fulfill the Bell-CHSH inequality
\[
\bigl\langle AB\bigr\rangle\mid_{0,0}+\bigl\langle AB\bigr\rangle\mid_{0,1}+\bigl\langle AB\bigr\rangle\mid_{1,0}-\bigl\langle AB\bigr\rangle\mid_{1,1}\;\leq\,2\textrm{.}
\]
\end{prop*}
\begin{proof}
For the conditional expectation we can write

\[
\bigl\langle AB\bigr\rangle\mid_{a,b}=\sum_{\lambda\in\varLambda}\sum_{s^{\mathfrak{a}}\in S^{\mathfrak{a}}}\sum_{s^{\mathfrak{b}}\in S^{\mathfrak{b}}}p_{\Lambda}(\lambda)p_{0}^{\mathfrak{a}}(s^{\mathfrak{a}})p_{0}^{\mathfrak{b}}(s^{\mathfrak{b}})\bigl\langle AB\bigr\rangle_{a,b,\lambda,s^{\mathfrak{a}},s^{\mathfrak{b}}}
\]
with
\[
\bigl\langle AB\bigr\rangle_{a,b,\lambda,s^{\mathfrak{a}},s^{\mathfrak{b}}}=\sum_{A,B\in\{-1,1\}}\sum_{t^{\mathfrak{a}}\in S^{\mathfrak{a}}}\sum_{t^{\mathfrak{b}}\in S^{\mathfrak{b}}}AB\cdot p^{\mathfrak{a}}(A,t^{\mathfrak{a}}\mid a,0,\lambda,s^{\mathfrak{a}})p^{\mathfrak{b}}(B,t^{\mathfrak{b}}\mid0,b,\lambda,s^{\mathfrak{b}})\textrm{,}
\]
which can be interpreted as the output product expectation value for
fixed $a,b\in\{0,1\}$; $\lambda\in\varLambda,s^{\mathfrak{a}}\in S^{\mathfrak{a}},s^{\mathfrak{b}}\in S^{\mathfrak{b}}$.
Reordering gives

\begin{equation}
\bigl\langle AB\bigr\rangle_{a,b,\lambda,s^{\mathfrak{a}},s^{\mathfrak{b}}}=\bigl\langle A\bigr\rangle_{a,\lambda,s^{\mathfrak{a}}}\bigl\langle B\bigr\rangle_{b,\lambda,s^{\mathfrak{b}}}\label{eq:factorization}
\end{equation}
with
\[
\bigl\langle A\bigr\rangle_{a,\lambda,s^{\mathfrak{a}}}=\left(\sum_{A\in\{-1,1\}}\sum_{t^{\mathfrak{a}}\in S^{\mathfrak{a}}}A\cdot p^{\mathfrak{a}}(A,t^{\mathfrak{a}}\mid a,0,\lambda,s^{\mathfrak{a}})\right)\textrm{,}
\]

\[
\bigl\langle B\bigr\rangle_{b,\lambda,s^{\mathfrak{b}}}=\left(\sum_{B\in\{-1,1\}}\sum_{t^{\mathfrak{b}}\in S^{\mathfrak{b}}}B\cdot p^{\mathfrak{b}}(B,t^{\mathfrak{b}}\mid0,b,\lambda,s^{\mathfrak{b}})\right)\textrm{.}
\]
These expressions are the output expectation values of $M^{\mathfrak{a}}$
and $M^{\mathfrak{b}}$ with fixed $a,\lambda,s^{\mathfrak{a}}$,
respectively $b,\lambda,s^{\mathfrak{b}},$ and lie in the interval
$[-1,1]$. So (according Appendix A) the absolute value of the expression
\[
E_{\lambda,s^{\mathfrak{a}},s^{\mathfrak{b}}}
\]
\begin{equation}
=\bigl\langle A\bigr\rangle_{0,\lambda,s^{\mathfrak{a}}}\bigl\langle B\bigr\rangle_{0,\lambda,s^{\mathfrak{b}}}+\bigl\langle A\bigr\rangle_{0,\lambda,s^{\mathfrak{a}}}\bigl\langle B\bigr\rangle_{1,\lambda,s^{\mathfrak{b}}}+\bigl\langle A\bigr\rangle_{1,\lambda,s^{\mathfrak{a}}}\bigl\langle B\bigr\rangle_{0,\lambda,s^{\mathfrak{b}}}-\bigl\langle A\bigr\rangle_{1,\lambda,s^{\mathfrak{a}}}\bigl\langle B\bigr\rangle_{1,\lambda,s^{\mathfrak{b}}}\label{eq:BellCHSHofProducts}
\end{equation}
will be less or equal $2$ for all $\lambda\in\varLambda,s^{\mathfrak{a}}\in S^{\mathfrak{a}},s^{\mathfrak{b}}\in S^{\mathfrak{b}}$.
Hence,
\[
\bigl\langle AB\bigr\rangle\mid_{0,0}+\bigl\langle AB\bigr\rangle\mid_{0,1}+\bigl\langle AB\bigr\rangle\mid_{1,0}-\bigl\langle AB\bigr\rangle\mid_{1,1}
\]
\begin{eqnarray*}
 & = & \bigl|\sum_{\lambda\in\varLambda}\sum_{s^{\mathfrak{a}}\in S^{\mathfrak{a}}}\sum_{s^{\mathfrak{b}}\in S^{\mathfrak{b}}}E_{\lambda,s^{\mathfrak{a}},s^{\mathfrak{b}}}p_{\varLambda}(\lambda)p_{0}^{\mathfrak{a}}(s^{\mathfrak{a}})p_{0}^{\mathfrak{b}}(s^{\mathfrak{b}})\bigr|\\
 & \leq & \sum_{\lambda\in\varLambda}\sum_{s^{\mathfrak{a}}\in S^{\mathfrak{a}}}\sum_{s^{\mathfrak{b}}\in S^{\mathfrak{b}}}\bigl|E_{\lambda,s^{\mathfrak{a}},s^{\mathfrak{b}}}\bigr|p_{\varLambda}(\lambda)p_{0}^{\mathfrak{a}}(s^{\mathfrak{a}})p_{0}^{\mathfrak{b}}(s^{\mathfrak{b}})\\
 & \leq & \sum_{\lambda\in\varLambda}\sum_{s^{\mathfrak{a}}\in S^{\mathfrak{a}}}\sum_{s^{\mathfrak{b}}\in S^{\mathfrak{b}}}2p_{\varLambda}(\lambda)p_{0}^{\mathfrak{a}}(s^{\mathfrak{a}})p_{0}^{\mathfrak{b}}(s^{\mathfrak{b}})=2\textrm{.}
\end{eqnarray*}
\end{proof}
\begin{example}
The machine functions of the FPSM pairs $(M_{1}^{\mathfrak{a}},M_{1}^{\mathfrak{b}})$
and $(M_{2}^{\mathfrak{a}},M_{2}^{\mathfrak{b}})$ in example 2 are
independent of both inputs. So they fulfill the Bell-CHSH inequality. 
\end{example}

\subsection{Bell's theorem }

The validity of the Bell-CHSH inequality is a logical consequence
of the functional independence of\emph{ }$p^{\mathfrak{a}}$ from
input $b$ and\emph{ }$p^{\mathfrak{b}}$ from input $a$. So the
violation of this inequality implies that there is some \emph{functional
dependence} instead: 
\begin{prop*}
For any pair of FPSMs $(M^{\mathfrak{a}},M^{\mathfrak{b}})$, defined
as above, which violates the Bell-CHSH inequality in the Bell test
simulation, the machine $M^{\mathfrak{a}}$ (the probabilistic machine
function $p^{\mathfrak{a}}$) depends on the selection input $b$
or the machine $M^{\mathfrak{b}}$ (the probabilistic machine function
$p^{\mathfrak{b}}$) depends on selection input $a$.
\end{prop*}
\noindent In this case the circuit diagram has to contain at least
one of the vertical connections (dotted lines in fig.\ \ref{fig:Simulation2FPSMsNoconn}).
\begin{example}
The FPSM pair $(M_{3}^{\mathfrak{a}},M_{3}^{\mathfrak{b}})$ in example
2 violates the Bell-CHSH inequality. The machine function $p_{3}^{\mathfrak{b}}=\delta_{B,1-2ab}$
depends on the input $a$. 
\end{example}

\subsection{Notes }
\begin{enumerate}
\item The expectation values in the RHS of (\ref{eq:factorization}) depend
only of one selection input $a$, respectively $b$. This has the
consequence that the expression (\ref{eq:BellCHSHofProducts}) has
only four variables, instead instead of eight. So the Bell-CHSH inequality
is fulfilled.
\item The proof will work even if we replace the product $p_{0}^{\mathfrak{a}}(s^{\mathfrak{a}})p_{0}^{\mathfrak{b}}(s^{\mathfrak{b}})$
with a joint probability distribution function $p_{0}^{\mathfrak{ab}}(s^{\mathfrak{a}},s^{\mathfrak{b}})$
on $S^{\mathfrak{a}}\times S^{\mathfrak{b}}$, where the initial state
distributions of the two machines may be not mutually independent
(this could be achieved by preprocessing of some $\lambda$-input).
Also in that case, the compound FPSM 
\[
M^{\mathfrak{ab}}=(\{0,1\}^{2}\times\varLambda,\{-1,1\}^{2},S^{\mathfrak{a}}\times S^{\mathfrak{b}},p_{0}^{\mathfrak{ab}},p^{\mathfrak{a}}p^{\mathfrak{b}})\textrm{.}
\]
fulfills the Bell-CHSH inequality. This shows that there is a significant
difference between the \emph{entanglement} of quantum states (which
can lead to a violation of the Bell-CHSH inequality in a similar situation,
cf.\ app.\ B) and the ordinary correlation of machine states (which
cannot).
\item With some additional measure theoretic assumptions the theorem can
also be proven for infinite systems. But this is more relevant for
physical theories then for automata theory.
\end{enumerate}

\section{Conclusion}

The Bell theorem for finite probabilistic sequential machines shows
that some data transmission between the separated machines is necessary
to reproduce the statistical results of the quantum-physical Bell
test experiment. The proof is simple and transparent. 

It sheds some light on the physical Bell theorem if we add some ontological
hypothesis, for example: any pair of locally separated physical systems
can be replaced by a pair of such machines. Then a ``spooky'' information
transmission over distances has to be assumed to explain the experimental
results (see \citealt{Gisin2008}). 

But the automata-theoretic version of the theorem has a value in itself.
It sets some limits for networks of probabilistic sequential machines
that are used for the description of communication devices. These
limits can be exceeded by quantum devices (cf.\ app.\ B), which
empowers quantum-cryptographic protocols (e.g., \citealp{Ekert1991}). 

\bibliographystyle{plainnat}

\appendix

\section{Bell-CHSH inequality}
\begin{prop*}
Any four real numbers $A_{0},A_{1},B_{0},B_{1}\in[-1,1]$ fulfill
the Bell-CHSH inequality

\begin{equation}
\bigl|A_{0}B_{0}+A_{0}B_{1}+A_{1}B_{0}-A_{1}B_{1}\bigr|\leq2\textrm{.}\label{eq:BellCHSH}
\end{equation}
\end{prop*}
\begin{proof}
The expression 
\[
A_{0}B_{0}+A_{0}B_{1}+A_{1}B_{0}-A_{1}B_{1}
\]
is linear in each of the four variables. So its maximum and minimum
are located on a corner of the hyper-cube $[-1,1]^{^{4}}$with $A_{0},A_{1},B_{0},B_{1}\in\{-1,1\}$.
In this case for all 16 possible valuations the following equation
is valid
\[
A_{0}B_{0}+A_{0}B_{1}+A_{1}B_{0}-A_{1}B_{1}=A_{0}\left(B_{0}+B_{1}\right)+A_{1}\left(B_{0}-B_{1}\right)=\pm2\textrm{.}
\]
\end{proof}
\begin{prop*}
Four random variables $A_{0},A_{1},B_{0},B_{1}\colon\varLambda\rightarrow[-1,1]$
on a measure space $(\varLambda,\Sigma)$ fulfill for any probability
measure $\mu$ on $(\varLambda,\Sigma)$ the Bell-CHSH inequality

\[
\bigl|\bigl\langle A_{0}B_{0}+A_{0}B_{1}+A_{1}B_{0}-A_{1}B_{1}\bigr\rangle_{\mu}\bigr|\leq2
\]
where $\bigl\langle\;\bigr\rangle_{\mu}$indicates the expectation
value with the measure $\mu$.
\end{prop*}
\begin{proof}
Let
\[
C(\lambda)=A_{0}(\lambda)B_{0}(\lambda)+A_{0}(\lambda)B_{1}(\lambda)+A_{1}(\lambda)B_{0}(\lambda)-A_{1}(\lambda)B_{1}(\lambda)\textrm{.}
\]
Then because of (\ref{eq:BellCHSH}) $\left|C(\lambda)\right|\leq2$
for all $\lambda\in\Lambda$ and 
\[
\bigl|\bigl\langle C\bigr\rangle_{\mu}\bigr|=\bigl|\intop_{\varLambda}C(\lambda)d\mu\bigr|\leq\intop_{\varLambda}\bigl|C(\lambda)\bigr|d\mu\leq\intop_{\varLambda}2d\mu=2\textrm{.}
\]
\end{proof}

\section{Example of quantum sequential machines violating Bell-CHSH inequality}

A quantum sequential machine can be defined in a similar way as a
probabilistic one. The essential difference is the use of complex-valued\emph{
amplitude} functions instead of non-negative real-valued probability
distribution functions. Calculations with these amplitudes are performed
in very a similar way as with probabilities. At the end of the calculation
the absolute (modulus) square of the resulting amplitude gives the
probability.

We define a finite quantum SM \emph{(FQSM)} as a quintuple $(I,O,S,\psi_{0},\varphi)$,
where $I,O,S$ are finite sets of input symbols, output symbols and
states\footnote{For quantum systems a more general notion of state is used. So we
should call $S$ more exactly the set of configurational states or
computational base states.},
\[
\psi_{0}\colon S\rightarrow\mathbb{C},s\mapsto\psi_{0}(s)
\]
is the \emph{initial state amplitude function} with 
\[
\sum_{s\in S}\bigl|\psi_{0}(s)\bigr|^{2}=1\textrm{,}
\]
and
\[
\varphi\colon O\times S\times I\times S\rightarrow\mathbb{C},(o,t,i,s)\mapsto\varphi(o,t|i,s)
\]
is the \emph{quantum machine function} with
\[
\sum_{o\in O}\sum_{t\in S}\bigl|\sum_{s\in S}\varphi(o,t|i,s)\psi_{0}(s)\bigr|^{2}=1
\]
for all $i\in I,s\in S$. 

The probability to get the output $o$ and the new state $t$\footnote{A \emph{measurement} has to be performed to get these results, but
we will not discuss this here (see for example \citealp{SayYakaryilmaz2014}).
We assume simply, that a measurement in the computational base is
performed after each input. } after the input $i$ in the initial state with amplitude $\psi_{0}$
is

\[
\bigl|\sum_{s\in S}\varphi(o,t|i,s)\psi_{0}(s)\bigr|^{2}\textrm{.}
\]
\begin{example}
Our example is a pair $(Q^{\mathfrak{a}},Q^{\mathfrak{b}})$ of FQSMs
that violates the Bell-CHSH inequality without dependence on the remote
input if it is initialized with a non-product initial state amplitude
function.
\begin{figure}[H]
\noindent \begin{centering}
\includegraphics[scale=0.5]{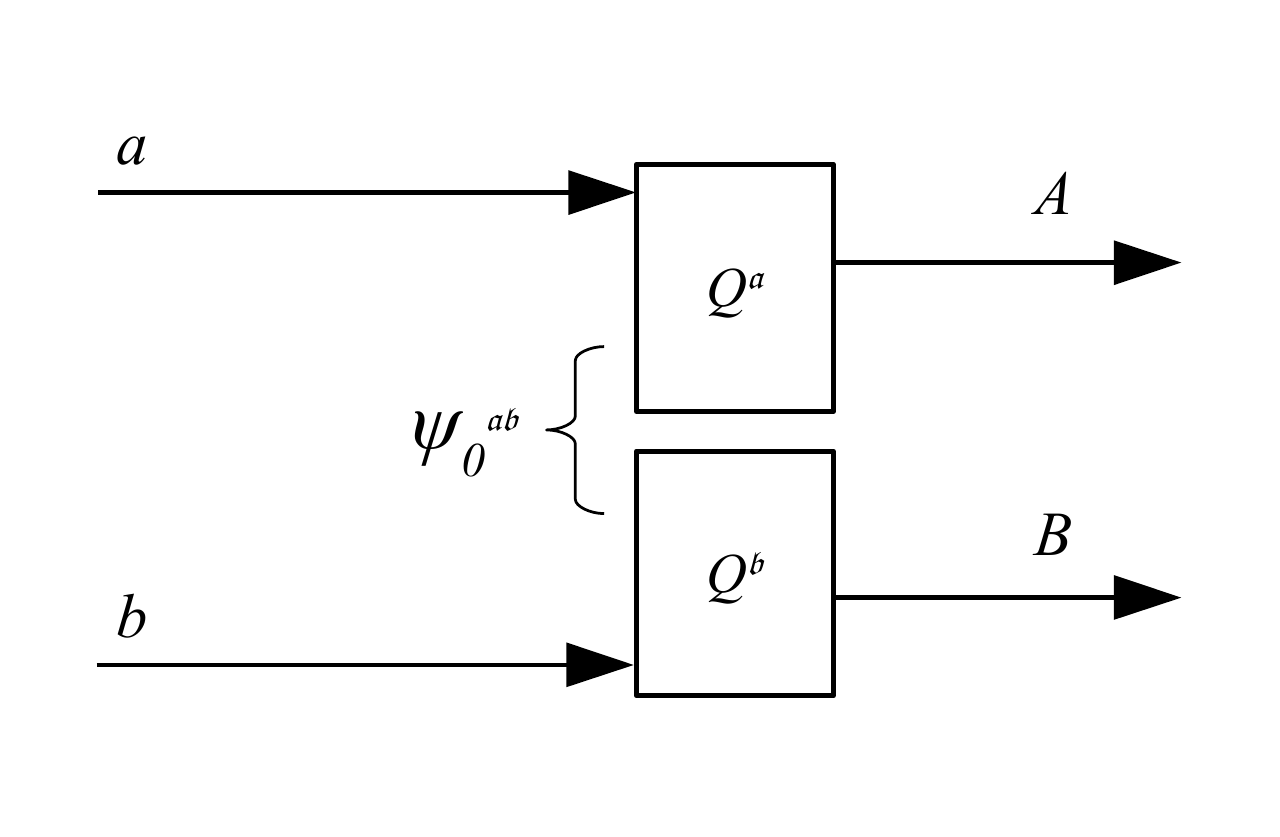}
\par\end{centering}
\caption{Simulation of the Bell experiment with two FQSMs with entangled initial
state amplitude function}
\end{figure}
 Both machines have the input set $I=\{0,1\}$ and the output set
$O=\{-1,1\}$. The state set is $S=\{0,1\}$\textemdash so both machines
are essentially \emph{qubit}s. The FQSMs have the form 
\[
Q^{\mathfrak{a}}=(\{0,1\},\{-1,1\},\{0,1\},\psi_{0}^{\mathfrak{a}},\varphi^{\mathfrak{a}})\textrm{,}
\]
\[
Q^{\mathfrak{b}}=(\{0,1\},\{-1,1\},\{0,1\},\psi_{0}^{\mathfrak{b}},\varphi^{\mathfrak{b}}).
\]
The quantum machine functions $\varphi^{\mathfrak{a}}$, $\varphi^{\mathfrak{b}}$
are (in table form):\\

\begin{tabular}{lcccc}
$\varphi^{\mathfrak{a}}$ & \multirow{1}{*}{$a=0$, $s^{\mathfrak{a}}=0$ } & $a=0$, $s^{\mathfrak{a}}=1$  & $a=1$, $s^{\mathfrak{a}}=0$  & $a=1$, $s^{\mathfrak{a}}=1$\tabularnewline
\midrule
$t^{\mathfrak{a}}=0$, $A=-1$ & $1$ & $0$ & $\frac{1}{\sqrt{2}}$ & $\frac{1}{\sqrt{2}}$\tabularnewline
$t^{\mathfrak{a}}=1$, $A=1$ & 0 & 1 & $\frac{1}{\sqrt{2}}$ & $-\frac{1}{\sqrt{2}}$\tabularnewline
\end{tabular}

\begin{tabular}{lcccc}
$\varphi^{\mathfrak{b}}$ & \multirow{1}{*}{$a=0$, $s^{\mathfrak{b}}=0$ } & $a=0$, $s^{\mathfrak{b}}=1$  & $a=1$, $s^{\mathfrak{b}}=0$  & $a=1$, $s^{\mathfrak{b}}=1$\tabularnewline
\midrule 
$t^{\mathfrak{b}}=0$, $A=-1$ & $-\frac{1}{\sqrt{4+2\sqrt{2}}}$ & $\frac{1+\sqrt{2}}{\sqrt{4+2\sqrt{2}}}$ & $\frac{1}{\sqrt{4+2\sqrt{2}}}$ & $\frac{1+\sqrt{2}}{\sqrt{4+2\sqrt{2}}}$\tabularnewline
$t^{\mathfrak{b}}=1$, $A=1$ & $\frac{1+\sqrt{2}}{\sqrt{4+2\sqrt{2}}}$ & $\frac{1}{\sqrt{4+2\sqrt{2}}}$ & $-\frac{1+\sqrt{2}}{\sqrt{4+2\sqrt{2}}}$ & $\frac{1}{\sqrt{4+2\sqrt{2}}}$\tabularnewline
\end{tabular}\\
\\
Both machines are \emph{Moore machines} where the new state $t^{\mathfrak{a}}$,
respectively $t^{\mathfrak{b}}$, determines the output (i.e., $A=2t^{\mathfrak{a}}-1$,
$B=2t^{\mathfrak{b}}-1$). We omitted rows which contain zeros only
(e.g., $t^{\mathfrak{a}}=0,$ $A=1$). 

The conditional output expectation is 
\[
\bigl\langle AB\bigr\rangle\mid_{a,b}=\sum_{A\in\{-1,1\}}\sum_{B\in\{-1,1\}}AB\cdot q\left(A,B\mid a,b\right)
\]
with the conditional probability to get the output $(A,B)$ after
input $(a,b)$

\[
q\left(A,B\mid a,b\right)=
\]
\[
\bigl|\sum_{s^{\mathfrak{a}}\in\{0,1\}}\sum_{s^{\mathfrak{b}}\in\{0,1\}}\sum_{t^{\mathfrak{a}}\in\{0,1\}}\sum_{t^{\mathfrak{b}}\in\{0,1\}}\varphi^{\mathfrak{a}}(A,t^{\mathfrak{a}}\mid a,s^{\mathfrak{a}})\varphi^{\mathfrak{b}}(B,t^{\mathfrak{b}}\mid b,s^{\mathfrak{b}})\psi_{0}^{\mathfrak{a}}(s^{\mathfrak{a}})\psi_{0}^{\mathfrak{b}}(s^{\mathfrak{b}})\bigr|^{2}\textrm{.}
\]
\\

However, to violate the Bell-CHSH inequality, the product $\psi_{0}^{\mathfrak{a}}(s^{\mathfrak{a}})\psi_{0}^{\mathfrak{b}}(s^{\mathfrak{b}})$
has to be replaced with a non-product (i.e.,\emph{ }entangled) initial
state amplitude function 

\[
\psi_{0}^{\mathfrak{ab}}(s^{\mathfrak{a}},s^{\mathfrak{b}})=\frac{1}{\sqrt{2}}(\delta_{s^{\mathfrak{a}},0}\delta_{s^{\mathfrak{b}},1}-\delta_{s^{\mathfrak{a}},1}\delta_{s^{\mathfrak{b}},0})\textrm{.}
\]
In this case the compound FQSM 
\[
Q^{\mathfrak{ab}}=(\{0,1\}^{2},\{-1,1\}^{2},\{0,1\}^{2},\psi_{0}^{\mathfrak{ab}},\varphi^{\mathfrak{a}}\varphi^{\mathfrak{b}})
\]
gives the following conditional output probability (in table form):\\

\begin{tabular}{lcccc}
$q\left(A,B\mid a,b\right)$ & $a,b=0,0$  & $a,b=0,1$  & $a,b=1,0$ & $a,b=1,1$ \tabularnewline
\midrule
$A,B=(-1,-1)$ & $\frac{2+\sqrt{2}}{8}$ & $\frac{2+\sqrt{2}}{8}$ & $\frac{2+\sqrt{2}}{8}$ & $\frac{2-\sqrt{2}}{8}$\tabularnewline
$A,B=(-1,1)$ & $\frac{2-\sqrt{2}}{8}$ & $\frac{2-\sqrt{2}}{8}$ & $\frac{2-\sqrt{2}}{8}$ & $\frac{2+\sqrt{2}}{8}$\tabularnewline
$A,B=(1,-1)$ & $\frac{2-\sqrt{2}}{8}$ & $\frac{2-\sqrt{2}}{8}$ & $\frac{2-\sqrt{2}}{8}$ & $\frac{2+\sqrt{2}}{8}$\tabularnewline
$A,B=(1,1)$ & $\frac{2+\sqrt{2}}{8}$ & $\frac{2+\sqrt{2}}{8}$ & $\frac{2+\sqrt{2}}{8}$ & $\frac{2-\sqrt{2}}{8}$\tabularnewline
\end{tabular}\\
\\
This is identical with
\[
q_{5}\left(A,B\mid a,b\right)=\frac{2-\sqrt{2}}{8}+\frac{\sqrt{2}}{8}(2\delta_{A,B}+2ab-4ab\delta_{A,B})
\]
from FPSM $M_{5}$ in example 1 and gives the Tsirelson bound $2\sqrt{2}$
as expectation value of the Bell-CHSH expression. 
\end{example}

\end{document}